\documentclass[11pt]{article}
\usepackage[margin=0.75in,footskip=0.25in]{geometry}

\usepackage{orcidlink}
\usepackage{amsmath}%
\usepackage{amsfonts}%
\usepackage{amssymb}%
\usepackage{graphicx}
\usepackage{xspace}
\usepackage{amsthm}

\newtheorem{theorem}{Theorem}
\newtheorem{lesson}{Lesson}
\newtheorem{problem}{Problem}
\newtheorem{corollary}{Corollary}
\newtheorem{consequence}{Consequence}
\newtheorem{definition}{Definition}
\newtheorem{counterexample}{Counterexample}
\newtheorem{observation}{Observation}
\newtheorem{question}{Question}

\def\({\big(}
\def\){\big)}

\newcommand{\tn}{\textnormal}

\newcommand{\hilbert}{\mathcal{H}}
\newcommand{\mc}[1]{\mathcal{#1}}

\newcommand{\wh}[1]{\widehat{#1}}

\newcommand{\wt}[1]{\widetilde{#1}}

\newcommand{\oper}[1]{\wh{\mathbf{#1}}}

\newcommand{\dofs}{d.o.f.s\xspace}

\newcommand{\vs}{\textit{vs.}\ }
\newcommand{\eg}{\textit{e.g.}\ }

\newcommand{\etc}{\textit{etc}}

\newcommand{\schrod}{Schr\"odinger}

\newcommand{\ket}[1]{|#1\rangle}
\newcommand{\braket}[2]{\langle#1|#2\rangle}

\newcommand{\q}{\boldsymbol{q}}
\newcommand{\p}{\boldsymbol{p}}
\newcommand{\operq}{\oper{q}}
\newcommand{\operp}{\oper{p}}

\def\sref #1{\S\ref{#1}}

\begin{document}

\title{Sentient observers and the ontology of spacetime}
\author{Ovidiu Cristinel Stoica\ \orcidlink{0000-0002-2765-1562}
\thanks{Dept. of Theoretical Physics, NIPNE---HH, Bucharest, Romania. \\
	Email: \href{mailto:cristi.stoica@theory.nipne.ro}{cristi.stoica@theory.nipne.ro},  \href{mailto:holotronix@gmail.com}{holotronix@gmail.com}}}
\date{\today}

\maketitle

\begin{abstract}
I show that, by the same criteria that led to Galilean and Special Relativity and gauge symmetries, there is no way to identify a unique set of observables that give the structure of space or spacetime. In some sense, space is lost in the state space itself. Moreover, the relationship between the observables and the physical properties they represent becomes relative. But we can verify that they are not relative, and the spacetime structure is unique. I show that this implies that not all structures isomorphic with observers can be observers, contradicting Structural Realism and Physicalism. This indicates a strong connection between spacetime and the sentience of the observers, as anticipated by some early contributors to Special and General Relativity.
\end{abstract}

\section{Introduction}
\label{s:intro}

In a quantum world, space is identifiable with a structure associated with the position observables. However, due to the huge symmetry of the state space, this identification cannot be done uniquely based only on the structures and relations. Observers are needed to interpret which of the observables represent positions. But the same symmetry of the state space implies that, for any observer, there are infinitely many structures identical to the observer's structure, and all of them would identify as position observables different isomorphic structures. From the point of view of Structural Realism, there is no preferred choice, leading to a principle of ``meta-relativity'' at the level of the state space itself. Then what breaks this huge unitary symmetry, reducing it to the space or spacetime symmetries? I prove that only observer-like structures associated with a particular choice of spacetime can be observers, otherwise their brains wouldn't be able to contain reliable information about the external world. All other observer-like structures can only be philosophical zombies. This has the surprising implication that the ontology of spacetime coincides with that of consciousness, refuting those theories of mind in which consciousness is reducible to structural or relational aspects of the observer's brain.

In Section \sref{s:lessons-relativity}, I revisit the lessons learned from the Principle of Relativity.
In Section \sref{s:lessons-quantum} I show that these lessons don't extend to the state space, because there are many structures isomorphic with spacetime, and there is no unambiguous way to deduce the physical meaning of the physical properties from the observables representing them.
In Section \sref{s:observers} I show that to solve these ambiguities and give meaning of the observables, observers are required.
But if we define observers by their structure and dynamics only, the ambiguity extends to the observers. In Section \sref{s:ambiguity-structural-realism} I show that the fact that we can know the physical properties of the external world refutes Structural Realism.
In Section \sref{s:sentient-spacetime} I explain the relation between sentient observers and spacetime.
In Section \sref{s:physicalism-vs-observers} I explain why this refutes Physicalism, and how several thinkers, including some early contributors to Special and General Relativity, anticipated the relation between sentience and spacetime.

\section{Lessons from the Principle of Relativity}
\label{s:lessons-relativity}

We use coordinate systems to prove theorems in Euclidean geometry since antiquity \cite{MerzbachBoyer2011AHistoryOfMathematics}. This, along with the development of cartography and the lessons on perspective learned by artists, revealed to us that geometric properties are independent of coordinates. Coordinates are very useful conventions, but they are conventions. The usage of coordinate systems to express geometric properties should be done with discernment, to make sure we don't take the conventional, the relative, as an intrinsic truth.

Gradually, we arrived at the discovery of symmetry groups and invariance.
It's not an overstatement to say that this revelation culminated with Felix Klein's \emph{Erlangen Program} \cite{klein1872ErlangenProgram}.
We finally came to the realization that geometry is the study of symmetry transformations. By choosing different symmetry groups and group representations, we can classify the previously known geometries and ``predict'' new ones, similar to how Mendeleev predicted new types of atoms.

Returning to Euclidean geometry, it applies to any structure that satisfies its axioms. It applies to both absolute and relative spaces. Geometry by itself is unable to distinguish between these two cases, because even if space were absolute, only its relational properties are captured by geometry.
It could be the case that the ontology of space has a way to distinguish a special point, an origin, or a special direction, or three special orthogonal directions that would give a preferred basis. But Euclidean geometry is blind to these possible additional properties.
In terms of forgetful functors between categories, the Euclidean space is an inner product vector space that forgot its origin, and also the space of triples from $\mathbb{R}^3$ that forgot the coordinate system, and so on. What do we learn from this?
\begin{lesson}
\label{lesson:euclid}
Even if space were absolute, anisotropic, with a preferred origin, or any other preferred structure, if all we can measure are distances and angles, we can treat it as relative.
\end{lesson}

Moving through space takes time, so in mechanics positions and angles can change with time.
This requires an enlargement of the symmetry group used in Euclidean geometry, and also of the space on which it acts, leading to Galileo's group. Galileo's group is also captured by the Erlangen program, but now it acts on four dimensions, three of space and one of time.

Newton believed that space and time are absolute, whereas Leibniz believed that they are relative. Newton's motivation was to find something absolute in the world, something that testifies for the existence of God. Ironically, Leibniz's motivation is related to his theological views.
But probably both agreed that Physics is about the truths that ignore the absolute space and time, whether or not they are absolute.
In Newtonian Physics, the truth about the nature of space and time remains transcendent. So, there is still peace between Newton and Leibniz, in the common ground provided by Galileo's Principle of Relativity. From this, we learn:
\begin{lesson}
\label{lesson:galileo}
Newtonian Mechanics is independent of whether space and time are relative or absolute, and of their ontology.
\end{lesson}

But Special Relativity revealed that, if we ignore gravity, the right extension of the Euclidean group is not Galileo's group, but the Poincar\'e group.
Space and time are not absolutely separated, the correct extended geometry is Minkowski's.

Minkowski's arguments that spacetime is relative, or better, relational, are right, of course.
But this won't stop people who want them to be absolute for various reasons to continue to believe this. For example, someone who needs ``spooky action at a distance'' (at the ontological level and not just in correlations) in their preferred interpretation of Quantum Mechanics may say that there is an absolute foliation of spacetime. A preferred foliation is required for example by the \emph{pilot-wave theory} \cite{Bohm1952SuggestedInterpretationOfQuantumMechanicsInTermsOfHiddenVariables} and the \emph{objective collapse theories} \cite{GhirardiRiminiWeber1986GRWInterpretation}. It's just that this foliation doesn't transpire in classical experiments, and neither in the quantum ones for that matter, but it can always be claimed that it's there, in the ontology that gives us local beables \cite{Bell2004SpeakableUnspeakable,Maudlin2019PhilosophyofPhysicsQuantumTheory}.
And when we're talking only about ontology, but nothing that can show up in experiments or be falsified, it's hard to argue. So we have to admit this as a possibility, no matter how dim it is:
\begin{lesson}
\label{lesson:minkowski}
Einstein's Special Relativity and the underlying Minkowski geometry are independent of whether space and time are absolute or relative.
\end{lesson}

Of course, this independence can be understood simply as saying that space and time are relative, since nothing contradicts this. But they could be absolute, and the equations and our measurements could still be the same, conspiring somehow to conceal the absoluteness of space and time.

The unification of space and time resulted in more evidence for relativity: the unification of other apparently separate things, like the momentum $3$-vector and energy as the components of a $4$-vector (with mass as the invariant length of this $4$-vector), the electric $3$-vector field and the magnetic $3$-pseudovector field as components of an antisymmetric $4$-tensor, and so on.

Things don't stop here, since the same symmetry principles, promoted from global to local, shape other sectors of physics.
Consider an extension of spacetime, whose extra dimensions are invisible not necessarily because they're curled in a too small compact circle or manifold, but because nothing changes in these directions.
The internal geometry of these extra dimensions has its own symmetry group -- the gauge group. 
Add a \emph{connection} -- a way to transport tangent vectors in a parallel way, so that we can define derivatives of the vector fields.
The symmetries of such a geometry are local gauge transformations.
The internal space is relative, so there is no absolute gauge (that is, no basis in the internal space), there are no absolute potentials. This didn't change with the discovery of the Aharonov-Bohm effect, despite a widespread view that it needs absolute potentials, because this effect is explained by connections or by the difference between potentials, and these are gauge invariant.
But, again, we can't disprove those thinking that potentials are fundamental for whatever metaphysical reasons or simply by not taking the geometric description seriously.
So, to indulge them too,
\begin{lesson}
\label{lesson:gauge}
Gauge theory is the same whether or not there is an absolute but unobservable choice of the gauge.
\end{lesson}

Let's find the common denominator of these lessons. We've seen, through the Erlangen Program and its local extension to gauge theory, that symmetry plays a central role in all of these lessons. Putting them all together,
\begin{lesson}
\label{lesson:symmetry}
The symmetries of the physical laws may be broken by the ontology, but if this doesn't transpire in the observations, it has no place in our descriptions of the laws.
\end{lesson}

Of course, if the ontology breaks the symmetry of the laws in a way that is manifest in the experiments, it means we have to update the theory by replacing the laws with those satisfying the broken symmetry.
Lesson \ref{lesson:symmetry} refers to the situation when no experiment contradicts the physical law.

Moreover, since the symmetry groups admit more representations, they can explain more things that we initially expected, in a unified way.
The Poincar\'e symmetry allows us to classify particles, and their free evolution equations, by spin and mass \cite{Wigner1959GroupTheoryAndItsApplicationToTheQuantumMechanicsOfAtomicSpectra,Bargmann1964NoteOnWignerTheoremSymmetryOperations}.
Other numbers associated with particles appear as the invariants of the gauge symmetries, giving us a more detailed classification of particles. This is how hadrons were classified, and new quarks and leptons were predicted, and later discovered \cite{georgi1982lie,MJDHamilton2017MathematicalGaugeTheoryStandardModel}.
This leads us to another lesson:
\begin{lesson}
\label{lesson:symmetry-central}
Symmetry plays a central role in physics. If not a fundamental role, at least a very efficient systematizing one.
\end{lesson}

These lessons add-up to a lesson in epistemic humility: 
\begin{lesson}[Structural Realism]
\label{lesson:SR}
Experiments and inferences from experiments give us access to the relations, but not to the true nature of the relata.
\end{lesson}

This may surprise those who expected that science gives us access to reality itself.
But let's revise what measurements do: they compare sizes with some reference sizes. For example, they compare lengths or distances with some unit lengths, and this is true for any measurement unit.
Experiments may also consist in counting, for example we can count how many particles of each type resulted from a collision. Or they may consist of detecting if an event happened or not. But in all cases, all we get are relations, not the nature of the relata.
So the fair and epistemically modest position is that of \emph{(Epistemic) Structural Realism} \cite{sep-structural-realism}: Science tells us about relations, but not about the nature of the relata.

And these relations are mathematizable, in fact mathematical structures are nothing but relations \cite{Gratzer2008UniversalAlgebra}. Counting particles gets us numbers. Ratios are numbers. We can also obtain non-numerical results, for example topological ones, but whatever relations there are, they turn out to be describable by mathematics. This makes mathematics very effective in physics, and from Lesson \ref{lesson:symmetry-central}, this effectiveness is highly amplified by the power of  symmetries.

From the raw data collected from experiments, we can guess more than just seems to be out there. But we can't read unequivocally what the universal laws are from these data, not even if there are such laws. Our theoretical models are a guesswork, we postulate general rules that fit the data, and then we make more experiments to falsify or corroborate these rules.
And this seems to work well.

\section{Do these lessons extend to the quantum world?}
\label{s:lessons-quantum}

The lessons from the previous section should apply to the quantum world as well.
Otherwise we could find quantum effects that break them. This not only didn't happen, but, as the classification of particles by the invariants of the representations of the symmetry groups showed, it powered Quantum Theory. In fact, the main ingredient of Quantum Field Theory (QFT) is the Poincar\'e symmetry, leading, together with the gauge symmetries, to the entire particle and interaction zoo \cite{Weinberg2005QuantumTheoryOfFields}.

One may think that, due to entanglement, relativity of simultaneity is violated, especially in the Einstein-Podolsky-Rosen experiment \cite{EPR35,Bohm1951TheParadoxOfEinsteinRosenAndPodolsky,Bell2004SpeakableUnspeakable}.
But in fact, no matter how entangled, the wavefunction is how it is due to the Poincar\'e symmetry itself \cite{Stoica2021WhyTheWavefunctionAlreadyIsAnObjectOnSpace}, and can be understood as an object on space, not just on the configuration space \cite{Stoica2019RepresentationOfWavefunctionOn3D,Stoica2023TheRelationWavefunction3DSpaceMWILocalBeablesProbabilities}. There is no ``spooky action at a distance'' \cite{Stoica2019RepresentationOfWavefunctionOn3D,Stoica2023TheRelationWavefunction3DSpaceMWILocalBeablesProbabilities}, unless it makes its way in, either by the wavefunction collapse, as in the objective collapse theory \cite{GhirardiRiminiWeber1986GRWInterpretation}, or by non-local interactions, as in the pilot-wave theory \cite{Bohm1952SuggestedInterpretationOfQuantumMechanicsInTermsOfHiddenVariables}.
But this may happen in a way that is not manifest in our experiments, and this is what Lesson \ref{lesson:minkowski} is about.
However, these interpretations weren't so far able to explain simple quantum phenomena beyond non-relativistic quantum mechanics with a fixed number of particles, let alone to produce a quantum theory of fields, despite decades of attempts \cite{Wallace2023TheSkyIsBlueAndOtherReasonsQuantumMechanicsIsNotUnderdeterminedByEvidence}. On the other hand, the Poincar\'e symmetry allowed the discovery of the first working formulations of relativistic quantum mechanics in a couple of years \cite{Dirac1928QuantumTheoryOfElectron} and of QFT during the next couple of decades. In fact it enforced it upon us.

However, we will soon see that, despite Lesson \ref{lesson:symmetry} and other lessons, the very existence of space or spacetime requires a symmetry breaking of the very large symmetry group of the state space, a symmetry breaking that neither the quantum structures nor their dynamics provide it!

To see this, recall that a quantum world is described by a unit vector $\ket{\psi(t)}$ in a Hilbert space $\hilbert$ (the state space).
The state vector $\ket{\psi(t)}$ evolves unitarily, governed by a {\schrod}-type equation
\begin{equation}
\label{eq:evol}
\ket{\psi(t)}=e^{-\frac{i}{\hbar}\oper{H}t}\ket{\psi(0)},
\end{equation}
where the Hamiltonian operator $\oper{H}$ is a Hermitian operator.

The Hilbert space is endowed with a \emph{tensor product structure} -- a decomposition as a tensor product of the Hilbert spaces for the elementary particles.
This tensor product structure is known to us from the empirical data, since it's determined by observables \cite{ZanardiLidarLloyd2004QuantumTensorProductStructuresAreObservableInduced}.

The observables represent physical properties, so they have associated physical meanings. The physical meaning of each observable follows from experiments.
The observables are represented by Hermitian operators on $\hilbert$.
The observables pertaining to a subsystem are also represented by Hermitian operators on $\hilbert$, with the additional constraint that they are independent of the states of other subsystems.

Among the observables, the position observables $\hat{x}$, $\hat{y}$, $\hat{z}$ are more fundamental, in the following sense. First, they are related to space. Second, the momenta $\hat{p}_x:=-i\hbar\frac{\partial}{\partial x}$, $\hat{p}_y:=-i\hbar\frac{\partial}{\partial y}$, $\hat{p}_z:=-i\hbar\frac{\partial}{\partial z}$ are defined as partial derivatives with respect to the spectra of the positions. Third, all other observables (except for the spin observables and the observables corresponding to internal degrees of freedom, which are independent of positions and momenta) are functions of the position and the momentum observables.

And here we have a problem. 

\begin{problem}
\label{problem:observables}
There is no unique way to identify the subsystems and the observables -- for example the position observables -- only from the mathematical structure of Quantum Theory.
\end{problem}

This problem is not noticed in the literature, so I will have to explain it and convince the reader of its existence and importance.

Suppose we have a complete set of commuting observables.
It can be chosen to contain the position observables of all particles.
For $n$ particles the position observables are
\begin{equation}
\label{eq:q-obs}
\operq:=\(\hat{q}_1,\hat{q}_2,\ldots,\hat{q}_{3n}\)=\(\underbrace{\hat{x}_1,\hat{y}_1,\hat{z}_1}_{\tn{particle }1},\underbrace{\hat{x}_2,\hat{y}_2,\hat{z}_2}_{\tn{particle }2},\ldots,\underbrace{\hat{x}_n,\hat{y}_n,\hat{z}_n}_{\tn{particle }n}\).
\end{equation}

Other observables, for spin and charges of the internal symmetries, are needed to obtain a CSCO, and may be included as additional observables. For simplicity I will ignore them, but we should remember that they exist and can be added at any time, as needed.

Let's collect the spectra of the operators $\hat{q}_j$ in a manifold $\mc{C}_{\operq}$, named in the following position \emph{parameter space}. Let $\q\in\mc{C}_{\operq}$, where
\begin{equation}
\label{eq:q-par}
\q=\(q_1,q_2,\ldots,q_{3n}\).
\end{equation}
The spin and internal \dofs can be included as additional parameters, or used to label different connected components of the parameter space $\mc{C}_{\operq}$.

Then, the wavefunction of the $n$ particles, with respect to the position basis, is obtained from the state vector $\ket{\psi}$ by
\begin{equation}
\label{eq:wf}
\psi(\q)=\braket{\q}{\psi}.
\end{equation}

The momentum observables are
\begin{equation}
\label{eq:p-obs}
\operp:=\(\hat{p}_1,\hat{p}_2,\ldots,\hat{p}_{3n}\).
\end{equation}

The momentum parameter space $\mc{C}_{\operp}$ is obtained by Fourier transforming the position space, and keeping the spin and internal \dofs,
\begin{equation}
\label{eq:p-par}
\p=\(p_1,p_2,\ldots,p_{3n}\).
\end{equation}

All other observables are functions of the position and momentum operators, $\oper{f}\(\operq,\operp\)$.

For a variable number of particles, as in the Fock representation in Quantum Field Theory, the manifolds $\mc{C}_{\operq}$ and $\mc{C}_{\operp}$ are the unions of the manifolds for fixed numbers of particles.

Evidently, the symmetries of space allow different choices of the axes and the origins determining the position observables. Changing the origin amounts to shifting the eigenvalues, and changing the axes amounts to replacing the position observables with linear combinations of the old position observables.
Similar for the momentum observables, whose eigenvalues depend of the relative motion of the inertial reference frames.
In the following we will assume that these equivalences go without saying.

Suppose we have identified what all observables mean physically.
If we apply a unitary transformation $\oper{S}$ to the observables, we obtain another set of observables.
If the unitary transformation doesn't correspond to a change of space coordinates, the resulting operators can't represent positions.

But any two observables $\oper{S}\hat{q}_j\oper{S}^\dagger$ and  $\oper{S}\hat{p}_k\oper{S}^\dagger$ are in the same relation with one another as the original $\hat{q}_j$ and  $\hat{p}_k$.
And for any observable defined as a function $\oper{f}(\operq,\operp)$, 
$\oper{S}\oper{f}(\operq,\operp)\oper{S}^\dagger$ is in the same relation with $\oper{S}\operq\oper{S}^\dagger$ and $\oper{S}\operp\oper{S}^\dagger$ as the original $\oper{f}(\operq,\operp)$ is with $\operq$ and $\operp$.

Then how can we distinguish, from the structural relations only, the observables $\operq$, $\operp$, and the other observables $\oper{f}(\operq,\operp)$, from any other choice of observables $\oper{S}\operq\oper{S}^\dagger$, $\oper{S}\operp\oper{S}^\dagger$, and $\oper{S}\oper{f}(\operq,\operp)\oper{S}^\dagger$?
It is not sufficient to merely call $\operq$ but not other $\oper{S}\operq\oper{S}^\dagger$ positions, since the relational structure gives all of these choices equal footing.
The operators $\oper{S}\operq\oper{S}^\dagger$ and $\oper{S}\operp\oper{S}^\dagger$ are functions of $\operq$ and $\operp$ as well, so they are as real as any operator $\oper{f}(\operq,\operp)$.
But what if the position observables $\operq$ and the transformed ones $\oper{S}\operq\oper{S}^\dagger$ have completely different nature, different ontology? Again, our lessons taught us that this is irrelevant to the observations. Only relations are supposed to matter.

We arrived at what seems to be a counterexample to the Lessons from Section \sref{s:lessons-relativity}, particularly Lesson \ref{lesson:SR} (Structural Realism):

\begin{counterexample}[To Structural Realism]
\label{counterexample:SR}
Something has to break the symmetry of the quantum structure, resulting in a preferred choice of the position observables.
\end{counterexample}

But if not the structure, then what makes the position observables special compared to other observables?
Even if the unitary transformation $\oper{S}$ commutes with the Hamiltonian, so that the dynamics has the same form as it does in the position basis $\(\ket{\q}\)_{\q\in\mc{C}_{\operq}}$, we arrived at the following question:
\begin{question}
\label{question:whence-space}
What gives the position observables $\operq$ their preferred role among all observables of the form $\oper{S}\operq\oper{S}^\dagger$, even assuming that $\oper{S}\oper{H}=\oper{H}\oper{S}$?
\end{question}

Some researchers believe that the space structure emerges uniquely from the Hamiltonian $\oper{H}$ in its basis-independent form, that is, from its spectrum (including multiplicities).
This is supposed to proceed in steps: first obtain the tensor product structure that corresponds to regions of space, then use additional information from the state vector to obtain distances \cite{CarrollSingh2019MadDogEverettianism,Carroll2021RealityAsAVectorInHilbertSpace}.
The tensor product structure is supposed to emerge, at least for finite-dimensional Hilbert spaces, from the spectrum and a condition of $k$-locality \cite{CotlerEtAl2019LocalityFromSpectrum}.
But in \cite{Stoica2022SpaceThePreferredBasisCannotUniquelyEmergeFromTheQuantumStructure} it was shown that for any structure that would lead to physically observable differences of the kind we observe, and for any invariant method to obtain it, there are infinitely many physically distinct solutions.
The tensor product structure, regardless of the additional conditions, not only doesn't emerge uniquely, but the number of additional continuous parameters needed grows exponentially with the number of subsystems (or regions of space) \cite{Stoica2024DoesTheHamiltonianDetermineTheTPSAndThe3dSpace}.
And even if we assume a tensor product structure as given, the description of the world is maximally ambiguous \cite{Stoica2023PrinceAndPauperQuantumParadoxHilbertSpaceFundamentalism}.

\section{Observers and the meaning of the observables}
\label{s:observers}

So how can we answer Question \ref{question:whence-space}?
As I mentioned, we know the meaning of the observables from experiments.
But to understand these meanings as extracted from experiments, we need to include into the equation the experimentalists, the observers.
This is a different problem than the measurement problem, which is also sometimes understood as necessitating the inclusion of the observers.
In the present case, the observers are needed to give physical meanings to the observables, and this would be necessary even in a classical world, where the symmetries are symplectomorphisms of the phase space (canonical transformations), instead of unitary symmetries.

If we take Lesson \ref{lesson:SR} into account, we will have to describe the observers only by their structure. And the entire structure is contained in the wavefunction $\psi_{O}$ of the observer,
\begin{equation}
\label{eq:wf-observer}
\psi_{O}(\q)=\braket{\q}{\psi_{O}}.
\end{equation}

Let us apply a unitary transformation $\oper{S}$.
The position operators will transform into the operators $\operq':=\oper{S}\operq\oper{S}^\dagger$, having eigenvectors $\ket{\q'}=\oper{S}\ket{\q}$, for all $\q\in\mc{C}_{\operq}$. The new basis is parametrized by another parameter space, $\q'\in\mc{C}_{\operq'}$, consisting of the eigenvalues of $\operq'$.

Then, in the new basis of eigenvectors of $\operq'$, $\(\ket{\q'}\)_{\q'\in\mc{C}_{\operq'}}$, the wavefunction has the form
\begin{equation}
\label{eq:wf-observer-like-structure}
\psi_{O}(\q')=\braket{\q'}{\psi_{O}}.
\end{equation}

Therefore, in general, an observer that looks like an observer in the basis $\(\ket{\q}\)_{\q\in\mc{C}_{\operq}}$, will, in general, not look like an observer in the new basis $\(\ket{\q'}\)_{\q'\in\mc{C}_{\operq'}}$. But sometimes other systems that are not necessarily observers may look like observers in the new basis.
Let's define such structures:

\begin{definition}
\label{def:observer-like-structure}
An \emph{observer-like structure} $O'$ with respect to a basis $\(\ket{\q'}\)_{\q'\in\mc{C}_{\operq'}}$ is a state that looks in the basis $\(\ket{\q'}\)_{\q'\in\mc{C}_{\operq'}}$ just like an observer $O$ in the position basis $\(\ket{\q}\)_{\q\in\mc{C}_{\operq}}$,
\begin{equation}
\label{eq:observer-like-structure}
\braket{\q'}{\psi_{O'}}=\braket{\q}{\psi_{O}}.
\end{equation}
\end{definition}

Let $\oper{S}$ be a unitary transformation that commutes with the Hamiltonian.
Let $O'$ be an observer-like structure in the basis $\(\ket{\q'}\)_{\q'\in\mc{C}_{\operq'}}$.
Then, Structural Realism implies the following
\begin{observation}[Relativity of structure]
\label{obs:relativity-of-structure}
If an observer-like structure $O'$ conducts experiments, the results appear in the basis  $\(\ket{\q'}\)_{\q'\in\mc{C}_{\operq'}}$ just like they would appear to an observer $O$ in the position basis $\(\ket{\q}\)_{\q\in\mc{C}_{\operq}}$.
In particular, $O'$ would interpret as the position basis $\(\ket{\q'}\)_{\q'\in\mc{C}_{\operq'}}$, and not $\(\ket{\q}\)_{\q\in\mc{C}_{\operq}}$.
\end{observation}

Observation \ref{obs:relativity-of-structure} leads to the following consequence:
\begin{observation}[Indeterminacy of spacetime]
\label{obs:indeterminacy-of-spacetime}
If we assume Structural Realism, there are as many choices for the space structure (and consequently for spacetime) as there are unitary transformations that commute with the Hamiltonian and don't correspond to changes of the reference frame.
\end{observation}

Observation \ref{obs:indeterminacy-of-spacetime} shows that the answer to Question \ref{question:whence-space} can't be given by assuming Structural Realism. Moreover, all our Lessons from Section \sref{s:lessons-relativity} seem to no longer apply.
Or, if we insist that they apply, we have to admit a huge number of alternative ways to choose the spacetime structure!

\section{The ambiguity of Structural Realism}
\label{s:ambiguity-structural-realism}

One may think that, while different observer-like structures can in principle identify different position operators, this identification is somehow physically the same. That is, maybe these position observables represent the same space but in a different reference frame, or maybe they can be related by a gauge or another kind of symmetry transformation, so that the physical world looks the same to two observer-like structures.
Now we will see that this is not the case, and the ambiguity from Observation \ref{obs:indeterminacy-of-spacetime} is maximal.

To see this, consider an observer $O$, and let $E$ be her environment, the rest of the world. Let the total state of the world at the time $t$ be $\ket{\psi}=\ket{\psi_\omega}\otimes\ket{\psi_\varepsilon}$ (what follows work for entangled states too). Suppose that there is a physical property $A$ of the environment, having a definite value $a$. That is, the total state $\ket{\psi}$ corresponding observable $\hat{A}$ is an eigenstate of $\hat{A}$ with the eigenvalue $a$, $\hat{A}\ket{\psi}=a\ket{\psi}$.
Suppose that the observer $O$ knows this, that is, the value $a$ is encoded in the observer's memory as the value of the property $A$ of the environment. The property $A$ can be any property of the environment, the position of an external object, its size, its color, or even a microscopic property like the spin of a Silver atom, or the fact that an atom is a Silver atom.
This is how it looks on the parameter space $\mc{C}$ identified by the observer as her position configuration space.

Let $\ket{\psi'}=\ket{\psi'_\omega}\otimes\ket{\psi'_\varepsilon}$ be another state, so that $\ket{\psi'_\omega}=\ket{\psi_\omega}$, but for the environment $\ket{\psi'_\varepsilon}$, the property $A$ has a different value $a'\neq a$. That is, $\hat{A}\ket{\psi'}=a'\ket{\psi}'$.
Both states $\ket{\psi}$ and $\ket{\psi'}$ are possible states of the world, but the world is in the state $\ket{\psi}$, not in $\ket{\psi'}$.
However, I will show that there is another parameter space $\mc{C}'$ on which the state $\ket{\psi}$ appears just like a state $\ket{\psi'}$ satisfying $\hat{A}\ket{\psi'}=a'\ket{\psi}'$ would appear on $\mc{C}$.

Let's consider first the particular case when for every eigenvalue $a$ of $\hat{A}$, $-a$ is also an eigenvalue, and any two eigenspaces can be mapped one into the other by a unitary transformation. That is, if they are finite-dimensional, they all have the same dimension.
In \cite{Stoica2023PrinceAndPauperQuantumParadoxHilbertSpaceFundamentalism} it was shown that in this case there is a unitary transformation $\oper{S}$ that commutes with $\oper{H}$ and interchanges $\ket{\psi_\omega}\otimes\ket{\psi_\varepsilon}$ and $\ket{\psi'_\omega}\otimes\ket{\psi'_\varepsilon}$, for some $\ket{\psi'_\omega}\otimes\ket{\psi'_\varepsilon}$ as above.
Then, the transformation $\oper{S}$ also interchanges the parameter spaces $\mc{C}$ and $\mc{C}'$, so that $\ket{\psi_\omega}\otimes\ket{\psi_\varepsilon}$ appears on $\mc{C}'$ just like $\ket{\psi'_\omega}\otimes\ket{\psi'_\varepsilon}$ appears on $\mc{C}$.
The proof was given for the \emph{standard model of quantum measurements} (see \eg \cite{Mittelstaedt2004InterpretationOfQMAndMeasurementProcess}, \S 2.2(b), and \cite{BuschGrabowskiLahti1995OperationalQuantumPhysics}, \S II.3.4), which can account for spin measurements with the Stern-Gerlach device, various photon counters and beam splitter experiments, photon polarization measurements \etc. (\cite{BuschGrabowskiLahti1995OperationalQuantumPhysics}, \S VII).

This result can be extended easily to any Hermitian operator $\oper{A}$.
Let us extend it to the case when $a'\neq -a$, but any two eigenspaces can still be related by unitary transformations. Then, we can find a basis in which $\oper{A}$ is diagonal, interchange the diagonal elements so that $a'$ is replaced by $-a$ but $a$ remains untouched, and return to the original basis. We will get another operator $\wh{\wt{A}}$ having the same eigenvectors as $\hat{A}$, so that if we know that a vector $\ket{\psi}'$ is an eigenvector of $\wh{\wt{A}}$ with the eigenvalue $-a$, we can determine that $\ket{\psi}'$ is also an eigenvector of $\hat{A}$, but with eigenvalue $a'$.

Let us now consider a general Hermitian operator $\hat{A}$. In a basis in which $\hat{A}$ is diagonal, we can decompose each eigenspace as a direct sum of subspaces that can be all related by unitary transformations. Then, all vectors in each such subspace are eigenvectors with the same eigenvalue. This can be done so that the projector on each of these subspaces is diagonal in the basis in which $\hat{A}$ is diagonal. After that, we can change the diagonal elements of $\hat{A}$ so that they differ for different subspaces. After returning to the original basis, we get a Hermitian operator $\wh{\wt{A}}$, so that if we know that a vector $\ket{\psi}'$ is an eigenvector of $\wh{\wt{A}}$ with the eigenvalue $-a$, we can determine that $\ket{\psi}'$ is also an eigenvector of $\hat{A}$, but with eigenvalue $a'$.

We obtained that for every state $\ket{\psi}\in\hilbert$ containing an observer $O$ who knows the value $a$ of a property $A$ of the environment $E$, and for any other possible value $a'$ of $A$, there is an observer-like structure $O'$ isomorphic with $O$, on a different parameter space $\mc{C}'$, so that the memory of $O'$ contains the information that the value of the property $A$ of the environment of $O'$ is $a$, but in fact it is $a'$.
Note that the property $A$ is represented by an operator $\hat{A}\(\q,\p\)$ on $\mc{C}$, but on $\mc{C}'$ it is represented by $\hat{A}':=\hat{A}\(\q',\p'\)$, where $\q$ represents the positions on $\mc{C}'$, and $\q$ the canonically conjugate momenta, as explained in Section \sref{s:lessons-quantum}. This is why the same vector $\ket{\psi}$ can be eigenvector of both $\hat{A}$ and $\hat{A}'$, but with different eigenvalues $a\neq a'$.
Therefore, we can apply Lemma 1 from \cite{Stoica2023PrinceAndPauperQuantumParadoxHilbertSpaceFundamentalism} to any observable $\hat{A}$.

This has the following consequence: if we choose randomly a parameter space $\mc{C}'$ so that
\begin{enumerate}
	\item 
$\ket{\psi}$ appears on $\mc{C}'$ as containing an observer-like structure $O'$ isomorphic with $O$,
	\item 
$O'$ (just like $O$) encodes in the structure of its brain the information that the property $A$ of the environment of $O'$ has a definite value $a$,
	\item 
on $\mc{C}'$, the property $A$ has a definite value,
\end{enumerate}
then the value of the property $A$ on $\mc{C}'$ can be any eigenvalue of $A$.
Moreover, for any observer $O$ on $\mc{C}$ as above, the parameter spaces $\mc{C}'$ on which there is an observer-like structure isomorphic with $O$ as above, are uniformly distributed. So a random observer-like structure $O'$ would know nothing about the value of the property $A$ of its environment!
And this applies to any property of the environment.

Let us summarize the result we obtained:
\begin{theorem}
\label{thm:conscious-observer}
If observers were reducible to structures, any property of the external world would be unknown to them.
\qed
\end{theorem}

But we do know many properties of the external world. Therefore,
\begin{consequence}
\label{consequence:zombies}
There are observer-like structures that are not observers.
\end{consequence}

Even if these observer-like structures have the same structure, and even if, in their proper position basis, the Hamiltonian has the same form as it does in the observer's position basis, not all these observer-like structures can actually be observers. If they could be, a random observer-like structure would not know the properties of its environment.

It also follows that
\begin{corollary}
\label{thm:no-structural-realism}
Structural realism is refuted.
\end{corollary}
\begin{proof}
If structural realism were true, observers would be reducible to their structure and the dynamics. But this would contradict Consequence \ref{consequence:zombies}.
\end{proof}

\section{Sentient observers and spacetime}
\label{s:sentient-spacetime}

Let's see if now we can answer Question \ref{question:whence-space}.
Theorem \ref{thm:conscious-observer} shows that if observers were reducible to their structure, even if we take dynamics into account, they wouldn't know the value of any physical property.
But we can know the values of the physical properties. This also means that we are able to give physical meaning to the operators representing observables. In particular, we give the physical meaning of positions to some of these operators, and not to others, even if the latter are unitarily equivalent with the former.

But can we identify the position observables uniquely?
The answer is yes, because otherwise, if there were an ambiguity, we wouldn't be able to do this. In general, for any observables that can be measured, even in principle, we can assign a unique physical meaning. Any ambiguity would make it impossible for that observable to be measured.
But are there properties that we can't measure? Obviously absolute positions and velocities, and absolute potentials. 
In general, we can't assign unambiguously  physical meaning to the operators that are not observable.
It follows that
\begin{theorem}
\label{thm:uniqueness-physical-meaning}
The parameter space determined by the existence of the observers is unique up to spacetime and local gauge symmetries.
This uniquely determines the meaning of the position operators $\operq$ and of their canonical conjugates $\operp$, and also of all other observables $\oper{f}\(\operq,\operp\)$, up to spacetime and local gauge symmetries.
\qed
\end{theorem}

In particular, there is a relation between the existence of sentient observers and a preferred spacetime structure:
\begin{corollary}
\label{thm:uniqueness-space}
In a quantum theory where spacetime is classical, the existence of the observers uniquely determines a spacetime.
\end{corollary}
\begin{proof}
Local gauge transformations act on a fiber bundle's fibers, and commute with the projector on the base manifold. Therefore, they don't change the position.
Then, the position observables are determined up to space symmetries, or up to spacetime symmetries in relativistic formulations.
\end{proof}

This gives the following answer to Question \ref{question:whence-space}. Recall that there is nothing in the structure or the dynamics that exhibits a preferred choice of the observables that identify space or spacetime as preferred structures of the theory. The needed breaking of the unitary symmetry of Quantum Theory comes from the fact that only some of the observer-like structures corresponding to different choices of the parameter space can be sentient.

Note that this result applies to Classical Physics too, because a classical system can be described by a state vector, and Hamilton's equations can be replaced by a {\schrod} equation, as Koopman and von Neumann showed \cite{Koopman1931HamiltonianSystemsAndTransformationInHilbertSpace,vonNeumann1932KoopmanMethod}. The allowed state vectors, representing unsuperposed classical states, form a basis of the Hilbert space, and the momentum operators commute with the position operators, but these differences don't affect the discussion from this article and the proof of Theorem \ref{thm:no-structural-realism}. The classical case can also be proved directly \cite{Stoica2023AskingPhysicsAboutPhysicalismZombiesAndConsciousness}. In particular it applies to Classical General Relativity, because it admits a classical Hamiltonian formulation \cite{ADM1962TheDynamicsOfGeneralRelativity}.

For Quantum Gravity, this implies the existence of a preferred spacetime structure, but one having multiple instances, leading to a version of the \emph{many-worlds interpretation} endowed with genuine probabilities (\cite{Stoica2023TheRelationWavefunction3DSpaceMWILocalBeablesProbabilities} ch. 8).

\section{Physicalism vs. sentient observers}
\label{s:physicalism-vs-observers}

The definition of materialism changed over time. From a Newtonian mechanistic, billiard ball universe (which Newton himself didn't take as supporting materialism), the goal post seemed to move. There was a time when ``magnetism'' was considered by some as immaterial, being associated with hypnosis and telepathy, so others said this was impossible. Maybe this is why Maxwell tried to explain it by a gear-like mechanism in the atoms of aether, and aether theories continued to be mainstream until Special Relativity was discovered. Decades later, the content of the universe was still artificially divided in ``matter'' and ``radiation'' or ``fields'', until Quantum Field Theory managed to give a unified picture. Gradually, the term ``materialism'' was replaced by ``physicalism'' or ``naturalism'', as if it would be unphysical or even unnatural or supernatural for consciousness to be fundamental. There was a time, before the discovery of the Minkowski spacetime, when a fourth dimension was considered a spiritualist speculation.

With such a shift in the definition of physicalism, we may wonder if there is a way to test it. To be a falsifiable hypothesis and not an always changing, always evasive claim, it needs to be pinned in a clear definition. Let's try:
\begin{definition}
\label{def:physicalism}
\emph{Physicalism} is the thesis that everything is made of stuff that doesn't have phenomenal powers. Consciousness is reducible to the structure and the dynamics of that stuff, which can be called ``matter''.
\end{definition}

The term ``phenomenal'' is understood in the sense of \emph{phenomenology}, ``the study of structures of consciousness as experienced from the first-person point of view'' \cite{sep-phenomenology}. That is, according to Definition \ref{def:physicalism}, there is no need for more than Structural Realism to account for consciousness.
Not all who see themselves as physicalists bother to clarify their position, but this definition is consistent with those given by people who are more precise \cite{sep-physicalism,Goff2017ConsciousnessAndFundamentalReality}, and it's captured in the idea that there are no \emph{philosophical zombies} identical to us but lacking sentient experience \cite{KirkSquires1974ZombiesVMaterialists,sep-zombies}. According to physicalism, either philosophical zombies don't exist, or we all are philosophical zombies \cite{Dennett1993ConsciousnessExplained}.

The difference between observers and observer-like structures that can't be observers is precisely the fact that someone can be an observer, but there is no ``what is like to be'' an observer-like structure that is not an actual observer, even though, according to Theorem \ref{thm:conscious-observer}, such structures exist.

Maybe Newton's intuition (see Section \sref{s:lessons-relativity}) that connected space and time with the existence of something that transcends and grounds the material world was not completely vacuous after all.

After the discovery of General Relativity but even before the discovery of Quantum Mechanics, Bertrand Russell, one of the most lucid and thoughtful exponents of skepticism and atheism, stated that materialism was refuted \cite{Russell1921AnalysisOfMind}, and concluded that there is a unique substance with sentient powers beyond both mind and matter \cite{Russell1927AnalysisOfMatter,Russell2021TheProblemsOfPhilosophy}. Similarly monistic conclusions were reached by James \cite{James1904AWorldOfPureExperience}, Mach \cite{Mach1914TheAnalysisOfSensationsAndTheRelationPhysicalPsychical}, Eddington \cite{Eddington1928NatureOfThePhysicalWorld}, Whitehead \cite{Whitehead1978ProcessAndReality}, and Weyl \cite{Weyl2021SpaceTimeMatter}.
Also see Auger \cite{Auger2021TheMethodsAndLimitsOfScientificKnowledge}.
Not to mention the founders of Quantum Theory.

In connection with relativity, Eddington writes (\cite{Eddington2017SpaceTimeAndGravitation}, p. 146-147)
\begin{quote}
And yet, in regard to the nature of things, this knowledge is only an empty shell—a form of symbols. It is knowledge of structural form, not knowledge of content. All through the physical world runs that unknown content, which must surely be the stuff of our consciousness.
\end{quote}

These proposals were motivated by the need to include sentient experience in a geometric or structural-realist world. The present article confirms their vision, by starting from the quantum world, and exploring how meaning is assigned to the quantum operators, in particular to identify the structure corresponding to spacetime. Theorem \ref{thm:no-structural-realism} led to the conclusion that consciousness is fundamental and accounts for the needed breaking of unitary symmetry.

Some of these proposals were motivated both by the inclusion of consciousness, and by the need to explain how our sense of time flowing, called by Einstein ``a stubbornly persistent illusion'', is compatible with the relativistic block universe.
Weyl ties consciousness and its experience of flow of time in the relativistic block universe to consciousness \cite{Weyl2021SpaceTimeMatter}.
Jeans elaborates, proposing that \cite{Jeans2020TheMysteriousUniverse}:
\begin{quote}
our consciousness is like that of a fly caught in a dusting-mop which is being drawn over the surface of the picture; the whole picture is there, but the fly can only experience the one instant of time with which it is in immediate contact [...]
\end{quote}

A profound analysis of the sense of the experience of time in the block universe, in particular of Weyl's proposal, was done by Petkov \cite{Petkov2013FromIllusionsToReality}.

This article doesn't clarify the relation between the sense of time flowing and the block universe. The flow of time seems, in relation to the block universe, similar to what the formulation of geometry in a particular basis is to their relational formulation based on invariants. A worm's eye view \vs God's eye view. Dual aspects of the same sentiential ontology?

The results from this article don't explain what consciousness is, or what its sentiential ontology is. It only shows that sentience is not reducible to structure and dynamics, and that it is required to ground all observable physical properties, endowing them with meaning.

\addcontentsline{toc}{section}{\refname}

\end{document}